\documentclass[11pt,a4paper]{article}

\usepackage[utf8]{inputenc}
\usepackage[T1]{fontenc}
\usepackage[english]{babel}
\usepackage{geometry}
\usepackage{amsmath,amssymb}
\usepackage{enumitem}
\usepackage{cite}
\usepackage[pdfencoding=auto,psdextra]{hyperref}
\usepackage{xcolor}
\hypersetup{
  pdftitle={On boundedness of solutions of three-state Moore--Greitzer compressor model
    with nonlinear proportional--integral controller for the surge subsystem},
  pdfauthor={Shiriaev, Freidovich, Shepeljavyi, Robertsson, Johansson},
  colorlinks=true, linkcolor=blue!70!black, citecolor=blue!70!black, urlcolor=blue!70!black
}
\newcommand{\di}{\displaystyle}

\newcommand{\trn}{^{\scriptscriptstyle T}}
\def\ru{\rule[-1pt]{0pt}{10pt}}
\newcommand{\al}{\alpha}
\newcommand{\om}{\omega}
\newcommand{\ep}{\varepsilon}
\newcommand{\de}{\delta}
\newcommand{\ga}{\gamma}
\newcommand{\be}{\beta}
\newcommand{\si}{\sigma}
\newcommand{\la}{\lambda}
\def\QED{~\rule[-0.5pt]{5pt}{5pt}\par\medskip}

\newtheorem{theorem}{\bf Theorem}
\newtheorem{statement}{\bf Statement}

\title{On boundedness of solutions of three-state Moore--Greitzer compressor model with nonlinear proportional--integral controller for the surge subsystem}

\author{%
  Anton Shiriaev\thanks{%
    Department of Engineering Cybernetics, NTNU,
    NO--7491 Trondheim, Norway. Email: \href{mailto:anton.shiriaev@ntnu.no}{\texttt{anton.shiriaev@ntnu.no}}}%
  \and
  Leonid Freidovich\thanks{%
    Department of Applied Physics and Electronics, Ume{\r a} University,
    SE--901 87 Ume{\r a}, Sweden. Email: \href{mailto:leonid.freidovich@umu.se}{\texttt{leonid.freidovich@umu.se}}}%
  \and
  Aleksandr Shepeljavyi\thanks{%
    Department of Theoretical Cybernetics,
    St.\ Petersburg State University, 199034 St.\ Petersburg, Russia.
    Email: \href{mailto:shepeljavyi@spbu.ru}{\texttt{shepeljavyi@spbu.ru}}}%
  \and
  \fbox{Anders Robertsson}\thanks{%
    Formerly with the Department of Automatic Control,
    Lund University, Sweden. Deceased on March 3, 2023.}%
  \and
  Rolf Johansson\thanks{%
    Department of Automatic Control, Lund University,
    SE--221 00 Lund, Sweden.
    Email: \href{mailto:rolf.johansson@control.lth.se}{\texttt{rolf.johansson@control.lth.se}}}%
}%

\begin{document}

\maketitle

\begin{abstract}
The work focuses on Lagrange stability of the origin for the three-state Moore–Greitzer compressor model in closed loop with a nonlinear PI controller, tuned only to stabilize a lower-dimensional invariant surge-dynamics subsystem.
The linearization of the system is not stabilizable but the static nonlinearity satisfies a sector condition, and together with a structural property of the stall-dynamics subsystem, this plays an essential role in the analysis.
The main contribution provides explicit conditions on the controller parameters together with analytical arguments that guarantee boundedness of all solutions of the closed-loop system.
The analysis employs a non‑standard application of circle-criterion-based arguments.
Together with the additional arguments developed in the work, this stability test also shows that the closed‑loop system is robust to certain perturbations and model uncertainties.
\end{abstract}

\begin{flushleft}
\textbf{Keywords:} Nonlinear systems; Circle criterion; Integral Quadratic Constraints;
3-state Moore--Greitzer compressor model.
\end{flushleft}

\section{Introduction}
The $3$-state Moore-Greitzer (3-MG) model of the compressor dynamics  \cite{Gre76,Gre81,Moo84,GreM86,MooG86,GraE99} is an inspiring nonlinear control system that has been approached and explored for several decades as a challenging bench-mark example for feedback controller design and for analysis of nonlinear closed-loop systems {
 \cite{Bot94,PadVEGG94,Jag95,KrsKK95,BanHM97,KrsFKP98,WilJ99,XiaB00,PadGE01,WilHJS02,MagP03,FonLPK04,HagMB04,ShiJR04,ShiJRF05,ChaB06,HelB07,BohG07,Hel07,BohG08,ShiJRF09} and many others}.  
The system dynamics have only three state variables, whose time evolution is determined by the following differential equations, see \cite[Eqs. (59), (60), and (61)]{MooG86}, 
\begin{eqnarray}
\label{MG3}
\hbox{$\frac{d}{dt}$} \phi&=& \hbox{$\frac32$} \phi-\psi +\hbox{$\frac12$}\left[1-(1+\phi)^3\right] -3 R(1+\phi)\label{dot_phi}\\
\hbox{$\frac{d}{dt}$}\psi&=& \hbox{$\frac{1}{\be^2}$}\left(\phi-u\right)\label{dot_psi}\\
\hbox{$\frac{d}{dt}$} R&=& -\sigma R^2-\sigma R\left(2\phi+\phi^2\right)\label{dot_R},
\quad R(0)\ge0
\end{eqnarray}
Here $\phi(t)$ and $\psi(t)$ denote deviations of the averaged flow and the total-to-static pressure-rise coefficients from their nominal values, respectively; $u(t)$ is defined by deviation of the coefficient of the inverse throttle characteristic function from a nominal value; $t$ is a scaled time measured in radians of travel of the compressor wheel; $\sigma$ and $\be$ are positive constant parameters. The dynamics of $R(t)$, known as stall, is often treated as a structured perturbation to Eqs.~(\ref{dot_phi})--(\ref{dot_psi}).

If $R(0)=0$, then Eq.~(\ref{dot_R}) implies that  the stall variable $R(t)$ is of zero value for all $t\ge0$, $R(t)\equiv0$, and such property is in place irrespective of a control input applied to the system. 
In this case, the non-trivial part of 3MG model (\ref{dot_phi})-(\ref{dot_R}) is reduced to its subsystem 
\begin{equation}\label{surge_dynamics}
\begin{array}{rcl}
\hbox{$\frac{d}{dt}$} \phi&=& \hbox{$\frac32$} \phi-\psi +\hbox{$\frac12$}\left[1-(1+\phi)^3\right] \\[2mm]
\hbox{$\frac{d}{dt}$}\psi&=& \hbox{$\frac{1}{\be^2}$}\left(\phi-u\right)
\end{array}
\end{equation} 
commonly known as {\em surge dynamics}. The first non-linearity of Eq.~(\ref{dot_phi})  -- present as well in (\ref{surge_dynamics})
\begin{equation}
W^{\{\phi\}}(v)\triangleq1-(1+v)^3
\label{nonlinearity_in_surge}
\end{equation}
is the so-called {\em compressor characteristic.\/} 
The main contribution of this note is 
sufficient conditions for the family of the dynamic feedback controllers  
\begin{equation}\label{dynamic_state_feedback}
u =\phi - \be^2\left\{ \ru\la_\phi\phi+\la_\psi\psi+\la_z z+\al\left[1-(1+\phi)^3\right]\strut\right\},\qquad
\hbox{$\frac{d}{dt}$} z = -\phi.
\end{equation}
with appropriately chosen constants $\la_\phi$, $\la_\psi$, $\la_z$ and $\alpha$, which ensure {\em  boundedness of solutions of the closed-loop system} (\ref{dot_phi})-(\ref{dot_R}), (\ref{dynamic_state_feedback}). 
{This feedback includes PID action in $\phi$ and proportional actions in $\psi$ and in the static nonlinearity.}

As in the many of the previously reported contributions for the example, the result relies on robust stabilization of the surge subsystem (\ref{surge_dynamics}) followed by the analysis of the closed-loop system solution in presence of the stall dynamics.
Stabilization of the invariant subsystem (\ref{surge_dynamics}) of (\ref{dot_phi})-(\ref{dot_R}) is obviously necessary for solving the classical task of asymptotic stabilization of the origin of the full dynamics (\ref{dot_phi})-(\ref{dot_R}). 
However, as reported in  \cite{ShiFJRR:2010} and illustrated by various numerical studies, the quadratic stabilization of the surge subsystem (\ref{surge_dynamics}) is not sufficient for stabilization of the full dynamics (\ref{dot_phi})-(\ref{dot_R}) with non-trivial stall; and, even the boundedness of the closed-loop-system solutions requires a separate verification. 
The focus of the work is the sequence of arguments that establish the boundedness property. 
The rest of the paper is organized as follows: this section continues with preliminary results; the main contributions are presented in Section~\ref{main_results}; concluding remarks are given in Section~\ref{conclusions}; and the proofs of the technical results are collected in the Appendix.

\subsection{Preliminaries: on global stabilization of the surge subsystem (\ref{surge_dynamics}) by the controller (\ref{dynamic_state_feedback})}
\noindent
The surge dynamics (\ref{surge_dynamics}) augmented with (\ref{dynamic_state_feedback}) results in the closed-loop system re-written as
\begin{equation}
\hbox{$\frac{d}{dt}$}\left[ \begin{array}{c} \phi\\ \psi\\ z\end{array} \right]
= \underbrace{\left[\begin{array}{ccc} \frac32 & -1& 0\\ \la_\phi& \la_\psi& \la_z\\ -1 &0&0\end{array}\right]}_{\di ={\cal A}}
\left[\begin{array}{c} \phi\\ \psi\\ z\end{array}\right]+
\underbrace{\left[\begin{array}{c}\frac12\\ \al\\ 0
    \end{array}\right]}_{\di ={\cal B}}W^{\{\phi\}}(v), 
\quad
v=\underbrace{\left[\,1,\, 0,\, 0\,\right]}_{\di =
 {\cal C}}\left[\!\begin{array}{c} \phi\\ \psi\\ z \end{array}\!\right],
\label{clp_dynamic_state_feedback}    
\end{equation}   
where  the non-linearity $W^{\{\phi\}}(\cdot)$ is defined in (\ref{nonlinearity_in_surge}). 
Following  \cite{ShiFJRR:2010}, a set of coefficients for the stabilizing controller~\eqref{dynamic_state_feedback}
can be obtained by applying the Yakubovich Circle Criterion to the system~\eqref{clp_dynamic_state_feedback},
where the nonlinearity $W^{\{\phi\}}(\cdot)$ is constrained by the sector inequality
\begin{equation}\label{QC_exsmple}
- v\cdot W^{\{\phi\}}(v) -\hbox{$\frac34$}\, v^2 =-v\left(1-(1+v)^3\right)-\hbox{$\frac34$}\, v^2
=v^2\left(\hbox{$\frac32$}+v\right)^2\ge 0,
\end{equation}
which holds for all $v\in\mathbb{R}$.
Next, we state the criterion as applied to the system~\eqref{clp_dynamic_state_feedback}.\medskip

\begin{statement}\label{statement1}
The constants $\la_\phi$, $\la_\psi$, $\la_z$ and $\al$ that satisfy the following constraints
\begin{enumerate}
\item The pair $\left({\cal C},{\cal A}\right)$ is observable and $\left({\cal A},{\cal B}\right)$ is controllable and
the non-strict inequality 
\begin{equation}
-\hbox{Re}\left\{T(j\om)\ru \right\}-\hbox{$\frac34$}|T(j\om)|^2 
\le 0
\label{FC_dynamics_state_feedback}
\end{equation}
holds for any $\om\ge0$, where 
\begin{equation}\label{T_condition_example}
T(s)={\cal C}\left( sI_3-{\cal A}\ru\right)^{-1}{\cal B}=
\di\frac{\frac12 s^2-\left(\al+\frac{1}{2}\la_\psi\right)s}
{s^3-\left(\la_\psi+\frac32\right)s^2+\left(\la_\phi+\frac32\la_\psi\right)s-\la_z} 
\end{equation}
\item 
The  matrix 
$\left({\cal A}-\hbox{$\frac34$}{\cal B}{\cal C}\right)$
is Hurwitz, 
\end{enumerate}
are the parameters of the controller (\ref{dynamic_state_feedback}), for which
the origin of the surge subsystem  (\ref{surge_dynamics}) with the  feedback controller (\ref{dynamic_state_feedback}), i.e.\/ the system (\ref{clp_dynamic_state_feedback}), is  globally asymptotically stable.\hfill\QED 
\end{statement}
{\em Proof:\/} The assumptions imply that there is  a 
positive definite matrix $P$ which satisfies the inequality
$$
 X\trn P\left[ {\cal A} X+{\cal B} w\right] 
+ \left[-X\trn {\cal C}\trn w -\hbox{$\frac34$} X\trn {\cal C}\trn {\cal C} X  \right]\le 0, 
\quad \forall\, X\in {\mathbb R}^3,\, \forall\, w\in {\mathbb R}^1. 
$$
Hence, the time derivative of the Lyapunov function candidate $V(X)\triangleq\frac12 X\trn P X$ with such $P$ and $X\triangleq\left[\phi;\psi;z\right]$ evaluated along a solution of  (\ref{clp_dynamic_state_feedback}) satisfies the inequality
\begin{eqnarray}
\hbox{$\frac{d}{dt}$} V(X(t)) 
&=&    X(t)\trn P\left[ {\cal A} X(t)+{\cal B} W^{\{\phi\}}\left({\cal C} X(t)\right) \right]\label{dV<0}\\
&\le&  X(t)\trn P\left[ {\cal A} X(t)+{\cal B} W^{\{\phi\}}\left({\cal C} X(t)\right) \right]  
+
\left[-v(t) W^{\{\phi\}}\left(v(t)\right) -\hbox{$\frac34$} v^2(t)\right]
\,\le\, 0\nonumber
\end{eqnarray}
Therefore, the origin  of the closed-loop system (\ref{clp_dynamic_state_feedback})   is a Lyapunov stable equilibrium, but not necessarily asymptotically. Furthermore, any solution of (\ref{clp_dynamic_state_feedback}) is bounded and has an $\om$-limit set $\ga$, which is non-empty, compact and invariant for  (\ref{clp_dynamic_state_feedback}). 
To prove that the origin is attractive for any solution of the closed-loop system (and, therefore, it is globally asymptotically stable),  consider a solution $X^0(t)=[\phi^0(t);\psi^0(t);z^0(t)]$ of (\ref{clp_dynamic_state_feedback}) originated on $\ga$. 
For this solution $\frac{d}{dt}V(X^0(t))$ is zero, and the value of $V(X^0(t))$ is constant for all $t\ge 0$. Then, integrating  (\ref{dV<0})  along the solution $X^0(t)$ over a time-interval $[0,\tau]$, $\tau>0$, we conclude that $\forall\,\tau>0$  
$$
V(X^0(\tau)) - V(X^0(0)) +\!
\int\limits_0^\tau \!\left[-v^0(t) W^{\{\phi\}}\left(v^0(t)\right) -\hbox{$\frac34$} |v^0(t)|^2\right] dt 
= \!\int\limits_0^\tau \!\left[\phi^{0}(t)\right]^2\left(\hbox{$\frac32$}+\phi^{0}(t)\right)^2 dt \le 0 
$$
The last inequality implies that the $\phi^0(t)$-component  of the solution $X^0(t)$, by necessity, should be a constant, which equals either  to $0$ or to $-3/2$. The last case is impossible, since the $z$-dynamics of (\ref{clp_dynamic_state_feedback}) with $\phi^0(t)\equiv(-3/2)$ implies:
$z^0(t) = (-3/2)\cdot t+z^0(0)$. It grows without bound, while, as proven, all solutions are bounded. In turn, if $\phi^0(t)\equiv 0$, then  the $\phi$-dynamics of  (\ref{clp_dynamic_state_feedback}) implies that $\psi^0(t)\equiv 0$. And both components $\phi^0(t)=\psi^0(t)\equiv0$ of the solution $X^0(t)$ substituted to $\psi$-dynamics, imply that
$z^0(t)\equiv0$. Therefore, the only solution of (\ref{clp_dynamic_state_feedback}) that belongs to $\ga$ is the origin.\hfill\QED

\subsection{Preliminaries: on finite-time escape for solutions of the closed-loop system (\ref{dot_phi})-(\ref{dot_R}), (\ref{dynamic_state_feedback})}
\noindent 
\begin{statement}\label{statement2}
Suppose the parameters $\la_\phi$, $\la_\psi$, $\la_z$ and $\al$ of the controller (\ref{dynamic_state_feedback}) are chosen as described in Statement~\ref{statement1}.   
Then, solutions of  the closed-loop system (\ref{dot_phi})-(\ref{dot_R}), (\ref{dynamic_state_feedback}) do not escape to infinity in finite time. Furthermore, the stall variable $R(t)$ enters the interval $[0,1]$ and stays therein for all sufficiently large time moments.\hfill\QED 
\end{statement}
{\em Proof:\/} 
If the parameters $\la_\phi$, $\la_\psi$, $\la_z$, and $\al$ of the controller (\ref{dynamic_state_feedback}) are chosen as described  in Statement~\ref{statement1}, then there exists a $3\times 3$ matrix $P=P\trn>0$ that satisfies the inequality (non-strict LMI)
\begin{equation}
x\trn P\left({\cal A}x + {\cal B}  w_1\right)+ \left[-x\trn {\cal C}\trn w_1-\hbox{$\frac34$} x\trn {\cal C}\trn {\cal C} x\right]
\le 0, \quad \forall\, x\in{\mathbb R}^3,\ \forall\, w_1\in{\mathbb R}^1
\label{are_1}
\end{equation}
where the matrices ${\cal A}$, ${\cal B}$, and ${\cal C}$ are defined in (\ref{clp_dynamic_state_feedback}), see Theorem~2.12  \cite{YakLG04}. 

Consider a solution $X(t)=[\phi(t);\psi(t);z(t);R(t)]$, $R(0)>0$, of the 
3-MG model (\ref{dot_phi})-(\ref{dot_R}) augmented with the controller (\ref{dynamic_state_feedback}).  
By smoothness of the right-hand side of (\ref{dot_phi})-(\ref{dot_R}), (\ref{dynamic_state_feedback}) this solution is well defined for some interval of time $[0,\tau)$ with $\tau>0$. 
Given $P>0$ as a solution of the LMI (\ref{are_1}), consider the function $V(X)=\frac12 X_s\trn P X_s$ defined as a quadratic form of the subset of components of the state vector $X$ that excludes the stall variable, i.e. $X_s\triangleq[\phi;\psi;z]$, so that  $X=[X_s;R]$.

Taking into account both the quadratic constraint and that the static nonlinearity $W^{\{\phi\}}(\cdot)$ satisfies (\ref{QC_exsmple})
$$
- v\cdot W^{\{\phi\}}(v) -\hbox{$\frac34$}\, v^2 =-v\left(1-(1+v)^3\right)-\hbox{$\frac34$}\, v^2
=v^2\left(\hbox{$\frac32$}+v\right)^2\ge 0, \quad\forall\, v\in{\mathbb R}^1,
$$
and the non-strict LMI (\ref{are_1}), the value of the time derivative of $V(X)=\frac12 X_s\trn P X_s$ along the solution $X=X(t)$ can be majorized from above as follows 
\begin{small}
\begin{eqnarray}
\hbox{$\frac{d}{dt}$}V\bigl(X(t)\bigr)&=& X_s(t)\trn P \dot X_s(t) 
= X_s(t)\trn P\left( {\cal A} X_s(t)  + {\cal B} W^{\{\phi\}}\left({\cal C} X_s(t)\right)+\left[\!\begin{array}{c}1\\[-1mm] 0\nonumber\\[-1mm] 0\end{array}\!\right]\left\{-3\cdot R(t)\cdot\left[1+\phi(t)\strut\right]\right\}\right)\nonumber\\
&\le& \overbrace{X_s(t)\trn P\left( {\cal A} X_s(t)  + {\cal B} W^{\{\phi\}}\left({\cal C} X_s(t)\right)\right)+
    \underbrace{\left[-X_s(t)\trn {\cal C}\trn W^{\{\phi\}}\left({\cal C} X_s(t)\right)-\hbox{$\frac34$} X_s(t)\trn {\cal C}\trn {\cal C} X_s(t)\right]}_{\di \ge 0, \hbox{ see Eq. (\ref{QC_exsmple})}}}^{\di \le 0, \hbox{ see Eq. (\ref{are_1})}}\nonumber \\[-2mm] 
    &&  \qquad\ +\,X_s(t)\trn P\left[\!\begin{array}{c}1\\[-1mm] 0\\[-1mm] 0\end{array}\!\right]\left\{-3\cdot R(t)\cdot\left[1+\phi(t)\strut\right]\right\}\nonumber\\
&\le& X_s(t)\trn P\left[\!\begin{array}{c}1\\[-1mm] 0\nonumber\\[-1mm] 0\end{array}\!\right]\left\{-3\cdot R(t)\cdot\left[1+\phi(t)\strut\right]\right\}\nonumber\\ 
&=&
-3\cdot R(t)\cdot X_s(t)\trn P\left[\!\!\begin{array}{c}\phi(t)\\[-1mm] 0\\[-1mm] 0\end{array}\!\!\right]
-3\cdot R(t)\cdot X_s(t)\trn P\left[\!\begin{array}{c}1\\[-1mm] 0\\[-1mm] 0\end{array}\!\right]
\end{eqnarray}
\end{small}

\noindent
It is readily seen that two terms in the last expression admit the bounds $\forall\, R\ge0$,  $X\in{\mathbb R}^3$
\begin{equation}
3 \cdot R\cdot  X_s\trn P\left[\!\!\begin{array}{c}\phi\\[-1mm] 0\\[-1mm] 0\end{array}\!\!\right]\le
\de_1\cdot R \cdot X_s\trn P X_s, \qquad
3\cdot R\cdot X_s\trn P\left[\!\begin{array}{c}1\\[-1mm] 0\\[-1mm] 0\end{array}\!\right]\le 
\de_2 \cdot X_s\trn P X_s+ \de_3\cdot R^2, \quad 
\end{equation}
with some positive constants $\de_1$, $\de_2$, $\de_3$.
Hence, $\hbox{$\frac{d}{dt}$}V\bigl(X(t)\bigr)$ satisfies the inequality
\begin{equation}
\hbox{$\frac{d}{dt}$}V\bigl(X(t)\bigr)\le \de_1\cdot R(t) \cdot V\bigl(X(t)\bigr)+ \de_2 \cdot V\bigl(X(t)\bigr)+ \de_3\cdot R^2(t), \quad t\in[0,\tau).
\label{proof7_dV1}
\end{equation}
In turn,  the variable $R(t)$ has the uniform upper bound for the time interval $[0,\tau)$
\begin{equation}
R(t)\le\max\{1,R(0)\}=R_{\max}.
\end{equation} 
Indeed, if for  $t\in[0,\tau)$ the non-negative value of $R(t)>1$, then the time derivative of  $R(t)$ 
\begin{eqnarray*}
\hbox{$\frac{d}{dt}$} R(t)&=& -\sigma R^2-\sigma R\left(2\phi+\phi^2\right)=- \sigma R^2 + \sigma R - \sigma R \left(1+\phi\right)^2\\
&=& -\si\cdot R(t)\cdot\bigl(R(t)-1\bigl)-\si\cdot R(t)\cdot \bigl(1+\phi(t)\bigl)^2 
\end{eqnarray*}  
is negative until $R(t)$ becomes less than $1$, and this happens irrespective of the time evolution of $\phi(t)$-variable on $[0,\tau)$.
Thus, the inequality (\ref{proof7_dV1}) implies the new one:
\begin{equation}
\hbox{$\frac{d}{dt}$}V\bigl(X(t)\bigr)\le \bigl[\de_1\cdot R_{\max} + \de_2\bigr] \cdot V\bigl(X(t)\bigr)+ \de_3\cdot R_{\max}^2, \quad t\in[0,\tau).
\end{equation}
This differential inequality implies that the function $V\bigl(X(t)\bigr)$ does not grow faster than exponential, i.e.\/
for some $\de_4>0$ the following inequality holds
\begin{equation}
V\bigl(X(t)\bigr)\le \bigl[V(0) + \de_4\bigr]\cdot e^{\left[\de_1\cdot R_{\max} + \de_2\right]\cdot t}, \quad t\in[0,\tau).
\end{equation}  
Hence, none of $[\phi;\psi;z]$-components of $X(t)$ has a finite escape time, while the $R$-component of the solution, as proved, is bounded, and, after transition, resides within the interval $[0,1]$.\hfill\QED   

\subsection{Preliminaries: on  structural properties of the full dynamics (\ref{dot_phi})-(\ref{dot_R}), (\ref{dynamic_state_feedback})}
\noindent 
\begin{statement}\label{Ricc_ineq_example}
Given the matrices ${\cal A}$, ${\cal B}$, ${\cal C}$ defined in (\ref{clp_dynamic_state_feedback}) with the parameters of the controller (\ref{dynamic_state_feedback}) that meet conditions of Statement~\ref{statement1},  then there are $k>0$ and a $3\times 3$ matrix $\bar P>0$ such that the next non-strict LMI
\begin{equation}
x\trn \bar P\left({\cal A}x + {\cal B}   w_1+{\cal B}_R  w_2\right)+ \left[-x\trn {\cal C}\trn w_1-\hbox{$\frac34$} x\trn {\cal C}\trn {\cal C} x\right] -k|w_2|^2
\le 0, \quad \hbox{with}\quad {\cal B}_R\triangleq\left[1;\, 0;\, 0\right]
\label{are_2}
\end{equation}
is valid for any $x \in {\mathbb R}^3$ and any scalars $w_1$, $w_2\in {\mathbb R}^1$.\hfill\QED
\end{statement}
The statement would be trivial, if the inequality (\ref{FC_dynamics_state_feedback}) for the transfer function (\ref{T_condition_example}) was strict. 
However, this is not the case and the new coupling term -- ``$x\trn \bar P {\cal B}_R w_2$'' present in the left-hand side of (\ref{are_2}) -- should satisfy certain matching conditions to ensure existence of a solution for the LMI (\ref{are_2}). 

\noindent
{\em Proof:\/}  By assumption, the pair $\left({\cal A},{\cal B}\right)$ is controllable, then the pair $\left({\cal A},[{\cal B},{\cal B}_R]\right)$ is controllable as well. 
Therefore, by Theorem~2.12 by  \cite{YakLG04}, there is a symmetric matrix $\bar P$ that is a solution of the inequality (\ref{are_2}), provided that the associated Frequency Condition is valid. 
For the considered here case, it means that the following inequality 
\begin{equation}
-\hbox{Re}\left\{\tilde v^*\tilde w_1\right\}-\hbox{$\frac34$}|\tilde v|^2-k|\tilde w_2|^2\le0
\label{FC_example_01}
\end{equation}
holds for all complex numbers $\tilde v$, $\tilde w_1$, $\tilde w_2$ and all real $\om\in {\mathbb R}^1$ that satisfy the linear relations
\begin{equation}
\tilde v={\cal C}\tilde x,\quad  j\om\cdot\tilde x={\cal A}\tilde x+{\cal B}\tilde w_1+{\cal B}_R\tilde w_2.
\label{FC_example_02}
\end{equation}
The left-hand side of (\ref{FC_example_01}) can be equivalently rewritten as a Hermitian form
\begin{equation}
\left[\begin{array}{c} \tilde w_1\\ \tilde w_2\end{array}\right]^*
\Pi(j\om)
\left[\begin{array}{c} \tilde w_1\\ \tilde w_2\end{array}\right]=
\left[\begin{array}{c} \tilde w_1\\ \tilde w_2\end{array}\right]^*
\left[\begin{array}{cc} \Pi_{11}(j\om)& \Pi_{12}(j\om)\\ \Pi_{12}^*(j\om)& \Pi_{22}(j\om) \end{array}\right]
\left[\begin{array}{c} \tilde w_1\\ \tilde w_2\end{array}\right]
\label{Pi_01}
\end{equation}
with the elements
\begin{eqnarray}
\Pi_{11}(j\om)&=& -\hbox{Re}\left\{T_1(j\om)\right\}-\hbox{$\frac34$}|T_1(j\om)|^2,\nonumber\\
\Pi_{12}(j\om)&=& -\hbox{$\frac12$}T_2(j\om)-\hbox{$\frac34$}T^*_1(j\om)T_2(j\om)\label{Pi_02}\\
\Pi_{22}(j\om)&=& -\hbox{$\frac34$}|T_2(j\om)|^2-k\nonumber
\end{eqnarray}
where $T_1(\cdot)$, $T_2(\cdot)$ are the transfer functions from inputs $w_1$, $w_2$ to the output $v$ respectively
\begin{equation}
T_1(s)={\cal C}\left( sI_3-{\cal A} \right)^{-1}{\cal B}=\frac{b_1(s)}{a(s)},\qquad 
T_2(s)={\cal C}\left( sI_3-{\cal A} \right)^{-1}{\cal B}_R=\frac{b_2(s)}{a(s)}.
\label{Pi_03}
\end{equation}
Due to the assumptions, $\Pi_{11}(j\om)$ is non-positive for $\om\in {\mathbb R}^1$. 
Hence, the Hermitian form (\ref{Pi_01}) is non-positive -- as equivalently requested in Eqn.~(\ref{FC_example_01}) -- if and only if the determinant of the $2\times 2$-matrix function $\Pi(s)$ satisfies the inequality
\begin{equation}
\det\bigl\{ \Pi(j\om)\bigr\}=\Pi_{11}(j\om)\Pi_{22}(j\om)-\left|\Pi_{12}(j\om)\right|^2\ge0,\quad \forall\, \om\in {\mathbb R}^1.
\label{det_Pi_01}
\end{equation}
The direct calculations show that 
\begin{equation}
\det\bigl\{ \Pi(j\om)\bigr\}=k\cdot\left[\hbox{Re}\left\{T_1(j\om)\right\}+\hbox{$\frac34$}\,|T_1(j\om)|^2\right]-\hbox{$\frac14$}|T_2(j\om)|^2%
=\frac{1}{|a(j\om)|^2}\left[k\cdot |q(j\om)|^2-\hbox{$\frac14$}|b_2(j\om)|^2 \right],
\label{det_Pi_02}
\end{equation}
where the polynomials $b_2(s)$, $a(s)$ are from (\ref{Pi_03}), and the stable polynomial $q(s)$ is defined as the result of the spectral factorization
\begin{equation}
|q(j\om)|^2=\hbox{Re}\left\{b_1(j\om)a(-j\om)\right\}+\hbox{$\frac34$}|b_1(j\om)|^2.
\label{q}
\end{equation}
Since $a(s)$ is the strictly Hurwitz polynomial, then  (\ref{det_Pi_01}) is equivalent to the inequality
\begin{equation}
k\cdot |q(j\om)|^2\ge \hbox{$\frac14$}\,|b_2(j\om)|^2, \quad \forall\, \om\in {\mathbb R}^1
\label{det_Pi_03}
\end{equation}
As shown below, it is fulfilled for some sufficiently large gain $k>0$, if the next conditions hold:
\begin{enumerate}
\item[\{1\}] $\deg\{ q(s)\}\ge\deg\{b_2(s)\}$;
\item[\{2\}] $q(0)=0$, the multiplicity of the root at $0$ is one and $|q(j\om)|^2>0$ for $\om\ne0$; 
\item[\{3\}] $b_2(0)=0$.
\end{enumerate}
Properties \{1\}-\{2\} follow from the features of the controlled surge subsystem described 
above
that imply solvability of non-strict LMI (\ref{are_1}).
Meanwhile, property \{3\} means that the transfer function $T_2(s)$ defined in (\ref{Pi_03}) is zero at $s=0$. Indeed, it equals to
\begin{small}
\begin{equation}
T_2(0)
=
{\cal C}\left( 0\cdot I_3-{\cal A} \right)^{-1}{\cal B}_R
 =-\bigl[1,\, 0,\, 0\bigr]\left[\begin{array}{ccc} 3/2 & -1& 0\\ \la_\phi&
      \la_\psi& \la_z\\ -1 &0&0\end{array}\right]^{-1}\left[\begin{array}{c} 1\\ 0\\ 0\end{array}\right]
 =-\bigl[1,\, 0,\, 0\bigr]\left[\begin{array}{ccc} q_{11} & *& *\\ *&
      *& *\\ * &*&*\end{array}\right]\left[\begin{array}{c} 1\\ 0\\ 0\end{array}\right]\\
\end{equation}%
\end{small}%
where $q_{11}$ is the first element of the adjoint matrix for ${\cal A}$ divided on $a(0)$. 
Since 
\begin{equation}
q_{11}=\frac{1}{a(0)}\det\left|\begin{array}{cc} \la_\psi& \la_z\\ 0 &0\end{array}\right|=0
\end{equation} is zero, then $T_2(0)$ equals to zero as well.

To complete the proof, let us consider the fraction of two polynomials $\di{|b_2(j\om)|^2}/{|q(j\om)|^2}$. 
Due to properties \{2\} and \{3\}, they have a common root at $\om=0$ and the denominator is positive for  $\om\neq0$. 
Then, due to property \{1\}, the maximum of this rational function is finite
\begin{equation}
\max\limits_{\om\in {\mathbb R}^1} \frac{|b_2(j\om)|^2}{|q(j\om)|^2}<M.
\end{equation}
Therefore, if one defines the gain $k$ equals to $4\cdot M$ or larger, then the inequality (\ref{det_Pi_03}) holds for all $\om$ and existence of  a solution $\bar{P}$ of the non-strict LMI (\ref{are_2}) is proved.
The fact that $\bar{P}$ is positive definite, follows from the assumptions that the pair $\left({\cal C},{\cal A}\right)$ is observable  and the matrix $\left({\cal A}-\hbox{$\frac34$}{\cal B}{\cal C}\right)$
is strictly Hurwitz.\hfill\QED

To take advantage of the inequality (\ref{are_2}), introduce the quadratic form 
\begin{equation}\label{Vbar}
\bar{V}(X)=\hbox{$\frac12$} X_s\trn \bar{P} X_s
\end{equation}
with $\bar{P}$ denoting a positive definite solution of~\eqref{are_2}, 
$X=[\phi;\psi;z;R]$ the state vector of the closed-loop system 
\eqref{dot_phi}–\eqref{dynamic_state_feedback}, 
and $X_s=[\phi;\psi;z]$ its components corresponding to the controlled 
surge subsystem.
The time derivative of $\bar{V}(X)$ computed along a solution $X(t)$ of the closed-loop system (\ref{dot_phi})-(\ref{dynamic_state_feedback})  with 
\begin{equation}
W_1(t)\triangleq W^{\{\phi\}}\left({\cal C} X_s(t)\right)\quad \hbox{and}\quad W_2(t)\triangleq\bigl[-3\cdot R(t)\cdot(1+{\cal C} X_s(t))\bigr]
\end{equation} 
takes the form
\begin{small}
\begin{eqnarray}
\hbox to1mm{$
\hbox{$\frac{d}{dt}$}\bar{V}\bigl(X(t)\bigr)= X_s(t)\trn \bar{P} \left({\cal A} X_s(t)+{\cal B}W_1(t)+{\cal B}_{R}W_2(t)\strut\right)$}&&\nonumber\\
&=& \underbrace{X_s(t)\trn \bar{P} \left({\cal A} X_s(t)+{\cal B}W_1(t)+{\cal B}_{R}W_2(t)\strut\right) + 
\left\{\strut \Bigl[-X_s(t)\trn {\cal C}\trn W_1(t)-\hbox{$\frac34$} X_s(t)\trn {\cal C}\trn {\cal C} X_s(t)\Bigr] -k\Bigl|W_2(t)\Bigr|^2\right\}}_{\di \le 0, \hbox{ see (\ref{are_2})}}\nonumber\\
&&\phantom{X_s(t)\trn \bar{P} \left({\cal A} X_s(t)+{\cal B}W_1(t)+{\cal B}_{R}W_2(t)\strut\right)}
-\, \biggl\{\biggr. \underbrace{\Bigl[-X_s(t)\trn {\cal C}\trn W_1(t)-\hbox{$\frac34$} X_s(t)\trn {\cal C}\trn {\cal C} X_s(t)\Bigr]}_{\di = \phi(t)^2\cdot\Bigl(\hbox{$\frac32$}+\phi(t)\Bigr)^2, \hbox{ see Eq. (\ref{QC_exsmple})}} -k\Bigl|W_2(t)\Bigr|^2\biggl. \biggr\}\nonumber\\
&\le& -\,\phi(t)^2\cdot\Bigl(\hbox{$\frac32$}+\phi(t)\Bigr)^2 + k\cdot\Bigl|3\cdot R(t)\cdot (1+\phi(t))\Bigl|^2. 
\end{eqnarray}
\end{small}

\noindent
Integrating the last inequality over the time-interval $[0,T]$, we obtain the inequality 
\begin{equation}
\di \bar{V}\bigl(X(t)\bigr)+\int\limits_0^{T}\Bigl\{\phi(t)^2\cdot\bigl(\hbox{$\frac32$}+\phi(t)\bigr)^2 
        - k\cdot\bigl|3\cdot R(t)\cdot (1+\phi(t))\bigr|^2\Bigr\}dt\,\le\, \bar{V}(X(0))
\label{dV_integrated_example}
\end{equation} 
that holds for $T>0$ along every solution $X(t)$ of the closed-loop of Eqs. (\ref{dot_phi})-(\ref{dynamic_state_feedback}).
The integral in the left-hand side of inequality (\ref{dV_integrated_example}) is not sign-definite. However,  its non-positive summand admits for solutions of the system of Eqs. (\ref{dot_phi})-(\ref{dynamic_state_feedback}) an alternative representation:
\begin{statement} 
Given a solution $X(t)=\left[\phi(t);\psi(t);z(t);R(t)\right]$ of the closed-loop system of Eqs. 
(\ref{dot_phi})-(\ref{dynamic_state_feedback}), the following relation  
\begin{equation}
-\int\limits_0^T \Bigl|3\cdot R(t)\cdot (1+\phi(t))\Bigl|^2 dt
=\frac{9}{2\si}\left[R^2(T)-R^2(0)\right]+\di 9\int\limits_0^T R^3(t) dt - 9\int\limits_0^T R^2(t) dt
\label{int_WR2}
\end{equation}
holds for all $T>0$.\hfill\QED
\label{lem04}
\end{statement}
{\em Proof:\/} The stall dynamics of Eq. (\ref{dot_R}) can be rewritten as follows
\begin{equation}
\hbox{$\frac{d}{dt}$} R= -\sigma R^2-\sigma R\left(2\phi+\phi^2\right)
=- \sigma R^2 + \sigma R - \sigma R \left(1+\phi\right)^2
\end{equation}
Hence
\begin{equation}
R(t)\hbox{$\frac{d}{dt}$} R(t)=-\si\cdot R^3(t)+\si R^2(t)
-\frac{\si}{9}\cdot \Bigl(3\cdot R(t)\cdot (1+\phi(t))\Bigl)^2.
\end{equation} 
Integrating the last expression over the time interval $[0,T]$, we obtain the identity
\begin{equation}
\int\limits_0^T R(t)\hbox{$\frac{d}{dt}$} R(t) dt +  \si\int\limits_0^T  R^3(t)dt
-\si \int\limits_0^T  R^2(t)dt = -\frac{\si}{9} \int\limits_0^T \Bigl(3\cdot R(t)\cdot (1+\phi(t))\Bigl)^2 dt, 
\end{equation}
which is equivalent to Eq. (\ref{int_WR2}).\hfill\QED

\noindent
Identity (\ref{int_WR2}) allows re-writing inequality (\ref{dV_integrated_example}) as 
\begin{eqnarray}
\hbox to100mm{$\di \bar{V}\bigl( X(T)\bigr)\,+\,\int\limits_0^{T} \phi(t)^2 \Bigl(\hbox{$\frac32$}+\phi(t)\Bigr)^2dt
\,+\, \frac{9 k}{2\si} \bigl[R(T)\bigr]^2\,+\,\di 9k\int\limits_0^T \bigl[R(t)\bigr]^3 dt \,-\, 9k\int\limits_0^T \bigl[R(t)\bigr]^2 dt$}&&\label{dV_integrated_example02} \\
&\le& \bar{V}\bigl(X(0)\bigr) + \frac{9 k}{2\si} \bigl[R(0)\bigr]^2  \nonumber
\end{eqnarray}
The positiveness of the quadratic form $\bar{V}(X)$ in inequality (\ref{dV_integrated_example02}) helps deriving consequently the next rough estimate, valid along every solution of the closed-loop system 
\begin{equation}\label{dV_integrated_example03}
\di \rho\cdot z(T)^2 - 9\cdot k\cdot \int\limits_0^T \bigl[R(t)\bigr]^2 dt \,\le\,  V\bigl(X(0)\bigr) 
+ \frac{9 k}{2\si}\bigl[R(0)\bigr]^2, \quad\forall\, T\ge 0
\end{equation}
with a positive constant $\rho$
. Indeed, the quadratic form $\bar{V}(X)$ admits the lower bound
\begin{equation}
\rho\cdot\bigl( \phi^2+\psi^2+z^2\bigr)\le \bar{V}(X)=\hbox{$\frac{1}{2}$}X_s\trn\bar{P}X_s, \quad \bar{P}>0,
\quad X_s=\bigl[\phi;\psi;z\bigr] 
\end{equation}
for some $\rho>0$. Then the inequality (\ref{dV_integrated_example03}) follows from (\ref{dV_integrated_example02}) if one drops in its left-hand side all non-negative items except $\rho\cdot z^2$. 

\section{Main result: on boundedness of solutions of the closed-loop system (\ref{dot_phi})-(\ref{dot_R}), (\ref{dynamic_state_feedback})}\label{main_results}
\noindent 
\begin{theorem}\label{Thm1} 
Consider the closed-loop system of Eqs. (\ref{dot_phi})-(\ref{dot_R}), (\ref{dynamic_state_feedback}) with the parameters of the controller (\ref{dynamic_state_feedback}) that meet conditions of Statement~\ref{statement1}.  Then, every 
solution of 
this system  is bounded.\hfill\QED
\end{theorem} 
{\em Proof:\/} 
Suppose on contrary that the system has an unbounded solution 
\begin{equation}
X^{(u)}(t)=\left[\phi^{(u)}(t);\psi^{(u)}(t);z^{(u)}(t);R^{(u)}(t)\right],
\end{equation} 
then the following -- step-by-step -- arguments are in place
\begin{enumerate}
    \item $X^{(u)}(\cdot)$ has no finite escape time, i.e.\/ it is well defined for any $t\ge 0$, see 
    Statement~\ref{statement2}.
    \item The stall component $R^{(u)}(\cdot)$ of $X^{(u)}(\cdot)$ is bounded, see 
    Statement~\ref{statement2}.
    \item The monotone function 
    \begin{equation}\label{R2toinf}
    \Phi(T)\triangleq\int\limits_0^T \bigl[R^{(u)}(t)\bigr]^2 dt \to +\infty\quad \hbox{as}\quad T\to+\infty
    \end{equation}
    is unbounded. 
\end{enumerate}
For the latter, indeed, since the function $\bar{V}(X)=\frac12 X_s\trn \bar P X_s$ is the quadratic form with positive definite matrix $\bar{P}$ and the stall variable is non-negative, $R(t)\ge0$, then all the terms in the left-hand side of Eq. (\ref{dV_integrated_example02}) except one are non-negative. This term is exactly the  scaled function $\Phi(\cdot)$. If it would be bounded, then inequality (\ref{dV_integrated_example02}) can be rewritten as 
    \begin{equation}
    \bar{V}\bigl( X^{(u)}(T)\bigr)\le  \bar{V}\bigl(X(0)\bigr) + \frac{9 k}{2\si} \bigl[R^{(u)}(0)\bigr]^2
    + 9k\underbrace{\int\limits_0^{+\infty} \bigl[R^{(u)}(t)\bigr]^2 dt}_{\di =\lim\limits_{T\to\infty} \Phi(T)}<+\infty, \quad \forall\, T\ge0
    \end{equation}
    and immediately implies that all the components of $X^{(u)}(\cdot)$ are bounded. Therefore, if the solution $X^{(u)}(\cdot)$ is unbounded, then, by necessity, the function $\Phi(T)$ grows to $+\infty$.   

\noindent
To continue the arguments, we formulate further properties of the stall variable $R(t)$ that -- as proven in Appendix -- are  valid for every solution of the closed-loop system of Eqs. (\ref{dot_phi})-(\ref{dynamic_state_feedback}). 
\begin{statement}\label{st21} 
Given a solution $X(t)=\left[\phi(t);\psi(t);z(t);R(t)\right]$ of the closed-loop system of Eqs. 
(\ref{dot_phi})-(\ref{dynamic_state_feedback}) with $R(0)>0$ and given a constant $\ep$, $0<\ep<1$, denote as 
\begin{equation}\label{al_I}
I(t)\triangleq \frac{ \frac{d}{dt}\al(t)}{\left[1+\si\al(t)\right]^{1-\ep}}\quad\hbox{with}\quad
\al(t)=R(0)\int\limits_{0}^t \exp\biggl( -\si\int\limits_{0}^s \left\{\phi^2(\tau)+2\phi(\tau)\right\}d\tau \biggr)  ds
\end{equation}
Then the following relations 
\begin{eqnarray}
\int\limits_{0}^T \bigl[R(t)\bigr]^2dt &\le& \frac{1}{\si\cdot \ep}\cdot\max\limits_{t\in[0,T]} \left\{ \strut I(t)\right\}\label{int_R2_2}\\
\ln \bigl(I(T)\bigr)&=&\ln R(0)+\si(\ep-1)\int\limits_0^T \!R(\tau)d\tau-\si \int\limits_0^T \phi^2(\tau)d\tau+2\si \left[z(T)-z(0)\strut\right] \label{lnI_1}
\end{eqnarray}
hold for $T>0$.\hfill\QED
\end{statement}

Using this property, we can continue as follows. 
\begin{enumerate}[resume]
\item For an unbounded solution $X^{(u)}(\cdot)$ the property (\ref{R2toinf}) and the inequality (\ref{int_R2_2}) imply that the function $I(\cdot)$, defined for that solution by Eqn.~(\ref{al_I}),  is unbounded as well, i.e.\/ there is a sequence of time moments 
\begin{equation}
0<t_1<t_2<\dots<t_n<\dots, \quad t_n\to+\infty
\end{equation}
such that 
\begin{equation}\label{I_to_infty}
    0<I(t_1)<I(t_2)<\dots<I(t_n)<\dots, \quad I(t_n)\to+\infty \quad\hbox{as}\quad n\to\infty
\end{equation}
\item Since $\ln(\cdot)$ is a monotonic function, then the property (\ref{I_to_infty})  implies the similar one for the logarithm of $I(\cdot)$. Namely,  
\begin{equation}\label{ln_I_to_infty}
    -\infty<\ln\bigl(I(t_1)\bigr)<\ln\bigl(I(t_2)\bigr)<\dots<\ln\bigl(I(t_n)\bigr)<\dots, \quad \ln\bigl(I(t_n)\bigr)\to+\infty 
\end{equation}
\item 
Since $R(t)\ge0$, $\si>0$ and $0<\ep<1$, then the equality (\ref{lnI_1}) reveals that the limit relation (\ref{ln_I_to_infty}) implies the new one
\begin{equation}
\lim\limits_{n\to +\infty} \ln\bigl(I(t_n)\bigr)=+\infty \quad\Rightarrow\quad  
\lim\limits_{n\to +\infty} \biggl[\frac{2}{1-\ep}\cdot z^{(u)}(t_n)- \int\limits_0^{t_n} \!R^{(u)}(\tau)d\tau   \biggr]=+\infty
\end{equation}
\item 
Since for all sufficiently large time moments the stall variable $R(t)$ resides in-between $0$ and $1$, see 
Statement~1, then $R^{(u)}(t)\ge \bigl[R^{(u)}(t)\bigr]^2$ for all $t\ge t_c$ and  
\begin{eqnarray}
\hbox to 20mm{$\di\lim\limits_{n\to +\infty} \biggl[\frac{2}{1-\ep}\cdot z^{(u)}(t_n)- \int\limits_0^{t_n} \!R^{(u)}(\tau)d\tau   \biggr]=+\infty$}&&\nonumber\\ 
&\Rightarrow&\quad 
\lim\limits_{n\to +\infty} \biggl[\frac{2}{1-\ep}\cdot z^{(u)}(t_n)- \int\limits_0^{t_n} \bigl[R^{(u)}(\tau)\bigr]^2 d\tau   \biggr]=+\infty \label{star}
\end{eqnarray}
\item 
The last conclusion, in particular, implies that 
\begin{equation}
z^{(u)}(t_n)\to +\infty\quad  \hbox{as}\quad n\to+\infty
\end{equation} 
and, therefore, starting with some index $n\ge N_c$ each element satisfies the inequality
\begin{equation}
  \frac{2}{1-\ep}\cdot \frac{9k}{\rho} <z^{(u)}(t_n).
\end{equation}
If so, then for $n \ge N_c$
\begin{equation}
\begin{array}{rcl}
\di
\frac{\rho}{9k}\cdot \bigl[z^{(u)}(t_n)\bigr]^2 - \int\limits_0^{t_n} \bigl[R^{(u)}(\tau)\bigr]^2 d\tau
&>&\di
\frac{\rho}{9k}\cdot \underbrace{\frac{2}{1-\ep}\cdot \frac{9k}{\rho}}_{\di < z^{(u)}(t_n)} \cdot\, z^{(u)}(t_{n})- \int\limits_0^{t_n} \bigl[R^{(u)}(\tau)\bigr]^2 d\tau\\
&=& \di \frac{2}{1-\ep}\cdot z^{(u)}(t_n)- \int\limits_0^{t_n} \bigl[R^{(u)}(\tau)\bigr]^2 d\tau   
\end{array}
\end{equation}
Therefore, taking into account the limit relation (\ref{star}), we can conclude  
\begin{equation}
\lim\limits_{n\to +\infty} \biggl[\frac{\rho}{9k}\cdot \bigl[z^{(u)}(t_n)\bigr]^2- \int\limits_0^{t_n} \bigl[R^{(u)}(\tau)\bigr]^2 d\tau   \biggr]\ge 
\lim\limits_{n\to +\infty} \biggl[\frac{2}{1-\ep}\cdot z^{(u)}(t_n)- \int\limits_0^{t_n} \bigl[R^{(u)}(\tau)\bigr]^2 d\tau   \biggr]
=+\infty
\end{equation}
However, the last limit relation is impossible! It contradicts the inequality (\ref{dV_integrated_example03}), established for all solutions of the closed-loop system (\ref{dot_phi})-(\ref{dot_R}), (\ref{dynamic_state_feedback}). The contradiction helps  conclude that every solution of the closed-loop system is bounded.\hfill\QED 
\end{enumerate}

\section{Concluding remarks}\label{conclusions}
\noindent
Some of remarks are in order:
\begin{enumerate}
    \item 
Since the discovery of the circle criterion by  \cite{Yak63,Roz63,Bon64}, its use has been mainly to prove asymptotic stability of nonlinear systems and to characterize robustness; see the historical notes in  \cite{Fradkov:2020} and the textbook treatments  \cite{Des75,Vid93,Kha02,YakLG04}. Examples of applications include  \cite{ArcK01,JohR02} and many others. At the same time, the criterion is also useful for other tasks. In particular, it can help to specify classes of perturbations of a nonlinear system that keep the ``frequency condition'' valid, and it can be used to establish a passivity‑like inequality in integral form with a quadratic certificate {\em \`a la} a storage function, which holds along solutions whether the system is stable or not.
    \item 
 This point is elaborated in the paper, where the circle criterion is used to solve an alternative task: it helps single out parameters of a nonlinear proportional–integral (PI) controller (\ref{dynamic_state_feedback}) that ensure boundedness of the solutions of the closed-loop system when it is applied to stabilize the three-state Moore–Greitzer compressor model (\ref{dot_phi})--(\ref{dot_R}). The closed-loop system (\ref{dot_phi})--(\ref{dot_R}), (\ref{dynamic_state_feedback}) contains several nonlinearities, but only one of them satisfies the infinite sector condition (\ref{QC_exsmple}), while the other nonlinearities cannot be readily characterized by integral quadratic constraints (IQCs). 
A notable aspect of this case study is that the result does not rely on a Lyapunov function or on Lyapunov-like stability theorems. We did not find a Lyapunov function for the closed-loop system, and we verified boundedness of solutions without it.
    \item 
Instead, we exploit the structure of the remaining nonlinearity \(3R(1+\phi)\) in the surge subsystem. This term enters only one of the system equations, reflects the physics of the process, and describes the nonlinear coupling between the stall dynamics and the deviations of the mean flow and pressure in the operating regime of the compressor. This observation suggested exploring the ``frequency condition'' and asking whether this nonlinearity can satisfy an IQC. This led to the non‑strict inequalities (\ref{FC_example_01})--(\ref{FC_example_02}), which hold for an appropriate choice of the parameter \(k\) in the new Lur'e–Riccati inequality (\ref{are_2}). Using this property, which holds irrespective of the stability of the nonlinear system and can be viewed as a matching condition that the 3‑state Moore–Greitzer compressor model satisfies automatically, we then proved boundedness of the closed‑loop system solutions, together with other properties of these solutions.
    \item 
 Overall, the proposed steps for controller design are intuitive and well motivated for the application. The controller is a modified PI controller that includes a copy of one of the nonlinearities in the proportional part. Furthermore, the analysis of the closed‑loop system does not rely on high‑gain constructions, which is beneficial for implementation.
    \item
    The established boundedness of the solutions of the closed‑loop system of Eqs. (\ref{dot_phi})--(\ref{dot_R}), (\ref{dynamic_state_feedback}) with controller parameters that satisfy the conditions of Statement~\ref{statement1} is critical for the analysis of the system. In fact, once this property is in place, the result can be complemented by an additional set of conditions that ensure asymptotic stability of the origin of the closed‑loop system. Those arguments lie outside the scope of the present paper and will be reported separately.
    \item 
Finally, we recall that numerical simulations reported in  \cite{ShiFJRR:2010} show that quadratic stabilization of the surge subsystem does not imply stability of the augmented system that includes the stall dynamics. In those simulations, local instability of the origin was observed together with bounded solutions that are not attracted to the origin. The present paper provides a formal proof that supports this observation.
  \end{enumerate}

\bibliographystyle{plain}

\appendix

\section{Proof of Complementary Statement~\ref{st21}}
\noindent
To prove the statement, let us first verify that the stall variable $R(t)$, being a solution of the differential equation (\ref{dot_R}), admits the representations
\begin{equation}\label{R=}
R(t)= \frac{ \frac{d}{dt}\al(t)}{1+\si\al(t)}
\qquad \biggl( \ \Rightarrow\ R(t)= \frac1\si\cdot\hbox{$\frac{d}{dt}$}\Bigl[\ln\bigl(1+\si\al(t)\bigr)\Bigr]\ \biggr) 
\end{equation}
where the function $\al(t)$ is defined by Eq.~(\ref{al_I}). Indeed, 
\begin{equation}\label{dR=}
\hbox{$\frac{d}{dt}$}  R(t)=  
\Biggl(\frac{ \frac{d^2}{dt^2}\al(t)}{1+\si\al(t)}\Biggr) \underbrace{\Biggl(\frac{\frac{d}{dt}\al(t)}{\frac{d}{dt}\al(t)}\Biggr)}_{\di =1} - \frac{ \si\bigl[\frac{d}{dt}\al(t)\bigr]^2}{\bigl[1+\si\al(t)\bigr]^2}
= R(t)\cdot \Biggl(\frac{\frac{d^2}{dt^2}\al(t)}{\frac{d}{dt}\al(t)}\Biggr)-\si\cdot \bigl[R(t)\bigl]^2 
\end{equation}
To meet the format of Eq.~(\ref{dot_R}) , we need to compute $\frac{d^2}{dt^2}\al(t)$ and $\frac{d}{dt}\al(t)$. To this end 
\begin{eqnarray}
\hbox{$\frac{d}{dt}$}\al(t)&=&  R(0) \exp\biggl( -\si\int\limits_{0}^t \left\{\phi^2(\tau)+2\phi(\tau)\right\}d\tau \biggr),\label{dal}\\
\frac{d^2}{dt^2}\al(t)&=& R(0) \exp\biggl( -\si\int\limits_{0}^t \left\{\phi^2(\tau)+2\phi(\tau)\right\}d\tau \biggr)\cdot
\Bigl(-\si\cdot \bigl[\phi^2(t)+2\phi(t)\bigr]\Bigr).\label{ddal}
\end{eqnarray}
Substituting these expressions into Eq. (\ref{dR=}), it takes the form
\begin{equation}
\hbox{$\frac{d}{dt}$}  R(t)= -\si\cdot \bigl[\phi^2(t)+2\phi(t)\bigr] \cdot R(t) -\si\cdot \bigl[R(t)\bigl]^2 
\end{equation}
and literally coincides with Eq. (\ref{dot_R}); hence, the validity of Eq. (\ref{dR=}) is established.  

\noindent
To verify the inequality (\ref{int_R2_2}), let us re-use the  form of the solution written as Eq. (\ref{dR=})
\begin{eqnarray*}
  \int\limits_{0}^T \bigl[R(t)\bigr]^2dt & = &  
  \int\limits_{0}^T \left[\frac{\dot\al(t)}{(1+\si\al(t))}\right]^2dt\,=\,
  \int\limits_{0}^T \frac{\dot\al(t)}{\left[1+\si\al(t)\right]^{1-\ep}}\frac{\dot\al(t)}{\left[1+\si\al(t)\right]^{1+\ep}}\, dt\nonumber\\
  &\le& \max\limits_{t\in[0,T]} \left\{ \frac{ \dot\al(t)}{\left[1+\si\al(t)\right]^{1-\ep}}\right\}  \int\limits_{0}^T \frac{\dot\al(t)}{\left[1+\si\al(t)\right]^{1+\ep}} dt\nonumber\\
   &=& \max\limits_{t\in[0,T]} \bigl\{  I(t)\bigr\}\cdot 
   \frac{1}{(-\ep)\cdot\si}\left.\frac{1}{\left[1+\si\al(t)\right]^{\ep}}\right|_{t=0}^{t=T} \nonumber\\
   &=& \max\limits_{t\in[0,T]} \bigl\{  I(t)\bigr\}\cdot
   \frac{1}{\ep\cdot\si} \biggl(\frac{1}{\left[1+\si\al(0)\right]^{\ep}} - \frac{1}{\left[1+\si\al(T)\right]^{\ep}}\biggr)
   \,\le\, \max\limits_{t\in[0,T]} \bigl\{  I(t)\bigr\}\cdot
   \frac{1}{\ep\cdot\si} \nonumber
\end{eqnarray*}
In the last inequality we have used the facts that $\al(0)=0$ and that $\al(t)$ is always non-negative.

\noindent
To validate the equality (\ref{lnI_1}), let us observe that the representation (\ref{R=}) leads to 
\begin{equation}\label{intR=}
\int\limits_0^T R(t)dt = \int\limits_0^T \frac1\si\hbox{$\frac{d}{dt}$}\Bigl[\ln\bigl(1+\si\al(t)\bigr)\Bigr]=
\frac1\si \Bigl[\ln\bigl(1+\si\al(T)\bigr) - \ln\bigl(1+\si\underbrace{\al(0)}_{\di =0}\bigr)\Bigr]= \frac1\si\ln\bigl(1+\si\al(T)\bigr),
\end{equation}
and, taking advantage of Eqs. (\ref{dal}) and (\ref{intR=}), consequently leads to 
\begin{eqnarray}
\ln\bigl(R(T)\bigr)&=& \left.\ln\Bigl(\hbox{$\frac{d}{dt}$}\al(t)\Bigr)\right|_{t=T}-\ln\bigl(1+\si\al(T)\bigr)\nonumber\\
\label{lnR=}&=& \underbrace{\ln\bigl(R(0)\bigr) -\si\int\limits_{0}^T \left\{\phi^2(t)+2\phi(t)\right\}dt}_{\di= \left.\ln\Bigl(\hbox{$\frac{d}{dt}$}\al(t)\Bigr)\right|_{t=T}} - \underbrace{\si  \int\limits_0^T R(t)dt}_{\di= \ln\bigl(1+\si\al(T)\bigr)}
\end{eqnarray}
According to the formulas (\ref{al_I}) and (\ref{R=}), 
we obtain
\begin{equation}
I(T)=\left[1+\si\al(T)\right]^{\ep}\cdot R(T)\qquad \hbox{with}\qquad \ep\in(0,1).
\label{I}
\end{equation}
Furthermore, its natural logarithm can be computed based on Eqs.~(\ref{intR=}) and (\ref{lnR=}) as
\begin{eqnarray}
  \ln\bigl(I(T)\bigr)&=& \biggl[\ep\cdot\ln\bigl(1+\si\al(T)\bigr)\biggr]+\biggl[\ln \bigl(R(T)\bigr) \biggr]   \nonumber\\
  &=& \hbox to30mm{$\biggl[\ep \si  \int\limits_0^T R(t)dt\biggr]$}+\biggl[\ln\bigl(R(0)\bigr) -\si\int\limits_{0}^T \left\{\phi^2(t)+2\phi(t)\right\}dt
  - \si  \int\limits_0^T R(t)dt\biggr] \nonumber\\
  &=&  \si(\ep-1)  \int\limits_0^T R(t)dt +\ln\bigl(R(0)\bigr) -\si\int\limits_{0}^T \phi^2(t)dt -\si\int\limits_{0}^T 2\phi(t)dt
\end{eqnarray}
To get Eq. (\ref{lnI_1}) from the last formula, we  observe that  the dynamics of $z$-variable for the closed-loop system result in 
\begin{equation} 
z(T)-z(0)=-\int\limits_0^T \phi(t)dt.
\end{equation} 
The proof is completed.\hfill\QED

\end{document}